\documentclass[letterpaper, 10 pt, conference]{ieeeconf}  

\IEEEoverridecommandlockouts                              
\usepackage{epsfig} 

\title{\LARGE \bf
Coupled Control Systems: Periodic Orbit Generation with\\ 
Application to Quadrupedal Locomotion
}

\usepackage{amsfonts}
\usepackage{amsmath} 
\usepackage{amsthm}
    \theoremstyle{plain}
\usepackage{mathrsfs}
\usepackage{amssymb}
\usepackage{mathrsfs}

\usepackage{color}
\usepackage{euscript}
\usepackage{url}
\usepackage{wrapfig}
\usepackage{array}
\usepackage{xspace}
\usepackage{comment}
\usepackage{titlecaps}
\Addlcwords{are or etc in}
\usepackage{siunitx}
\usepackage{cases}
\usepackage{mathtools}

\usepackage{style/amesstyle}
\usepackage{style/wma}

\author{Wen-Loong Ma, Noel Csomay-Shanklin and Aaron D. Ames
\thanks{This work is supported by NSF grant 1724464.}
\thanks{The authors are with the department of Mechanical Engineering, and Control+Dynamical Systems, California Institute of Technology, Pasadena, CA, USA. {\tt\small wma, noelcs, ames@caltech.edu}.}
}

\begin{document}

\maketitle
\thispagestyle{empty}
\pagestyle{empty}


\begin{abstract}
A robotic system can be viewed as a collection of lower-dimensional systems that are coupled via reaction forces (Lagrange multipliers) enforcing holonomic constraints. Inspired by this viewpoint, this paper presents a novel formulation for nonlinear control systems that are subject to coupling constraints via virtual ``coupling'' inputs that abstractly play the role of Lagrange multipliers. The main contribution of this paper is a process---mirroring solving for Lagrange multipliers in robotic systems---wherein we isolate subsystems free of coupling constraints that provably encode the full-order dynamics of the coupled control system from which it was derived. This dimension reduction is leveraged in the formulation of a nonlinear optimization problem for the isolated subsystem that yields periodic orbits for the full-order coupled system. We consider the application of these ideas to robotic systems, which can be decomposed into subsystems. Specifically, we view a quadruped as a coupled control system consisting of two bipedal robots, wherein applying the framework developed allows for gaits (periodic orbits) to be generated for the individual biped yielding a gait for the full-order quadruped. This is demonstrated through walking experiments of a quadrupedal robot in simulation and on rough terrains.
\end{abstract}

\section{Introduction} 
\label{sec:intro}

To achieve dynamic walking on high-dimensional robotic systems, hybrid zero dynamics (HZD) has proven to be a successful methodology as a result of its ability to make theoretic guarantees \cite{Westervelt2007a, Grizzle2014Models, ames14CLF} and yield walking for complex humanoids \cite{Sreenath2011, reher2016realizing}. 
The main idea behind this approach is that the full-order dynamics of the robot can be reduced to a lower-dimensional surface on which the system evolves. The system can then be studied via the low-dimensional dynamic representation and, importantly, guarantees made can be translated back to the full-order dynamics, i.e., periodic orbits (or walking gaits) in the low-dimensional system imply corresponding periodic orbits in the full-order system. The goal of this paper is to capture this dimension reduction in a more general context---that of \systems, which capture the ability to decompose a complex system into low-dimensional subsystems.

Another means of dimension reduction for robotic systems comes from isolating subsystems and coupling these subsystems at the level of reaction forces, i.e., Lagrange multipliers that enforce holonomic constraints. This is the idea underlying the highly efficient method for calculating the dynamics of robotic systems: Spatial vector algebra \cite{featherstone2014rigid}. For example, a double pendulum can be decomposed into two single pendula connected via a constraint at the pivot joint \cite{ganesh2007composition}. More generally, one can consider two equivalent ways of expressing the dynamics of a robotic system \cite{Murray1994mathematical}: 
\par\vspace{-4.3mm}{\small\begin{align*}
    \underbrace{D(q) \ddot{q} + H(q,\dot{q}) = u}_{\textrm{Full-Order Dynamics}} 
    \Leftrightarrow 
    \underbrace{\begin{dcases}
    D_i(q_i) \ddot{q}_i + H_i(q_i,\dot{q}_i) = u_i + J_{h_i}^\top \lambda  \\
    \qquad \qquad  \mathrm{s.t.} \qquad  h(q)= 0
    \end{dcases}}_{\textrm{Reduced-Order Coupled Dynamics}}
\end{align*}}%
\vspace{-3.4mm}\par\noindent
for $i = 1,2$, where $h$ is a coupling (holonomic) constraint that is enforced via the Lagrange multiplier $\lambda$ allowing for the higher-dimensional $q$ to be decomposed into lower-dimensional $q_i$, i.e., $q = (q_1,q_2)$. For example, a quadrupedal robot can be decomposed into two bipeds as in \figref{fig:frontPage}. Thus, if one can make guarantees on the reduced-order coupled systems, they can be translated to the full-order dynamics. 

The study of coupled dynamic and control systems has a long and rich history from which the framework presented in this paper has taken inspiration. The most prevalent example is that of multi-robot systems \cite{mesbahi2010graph}, and specifically the consensus problem \cite{ren2005survey}. 
Interconnected systems have also been well-studied \cite{antonelli2013interconnected}. 
In the context of mechanical and robotic systems on graphs, network synchronization has been considered \cite{chung2009cooperative}.
Port-Hamiltonian systems also capture the notion of coupling present in general mechanical systems \cite{van2014port}. 
Finally, in related work, the coordination of quadruped and human reaction forces has recently been studied \cite{hamed2019hierarchical}. 
While not explicitly discussed due to space constraints, many of these formulations fit within the general setting of coupled control systems presented here.

\begin{figure}[t!]
	\centering
	\includegraphics[width=0.43\textwidth]{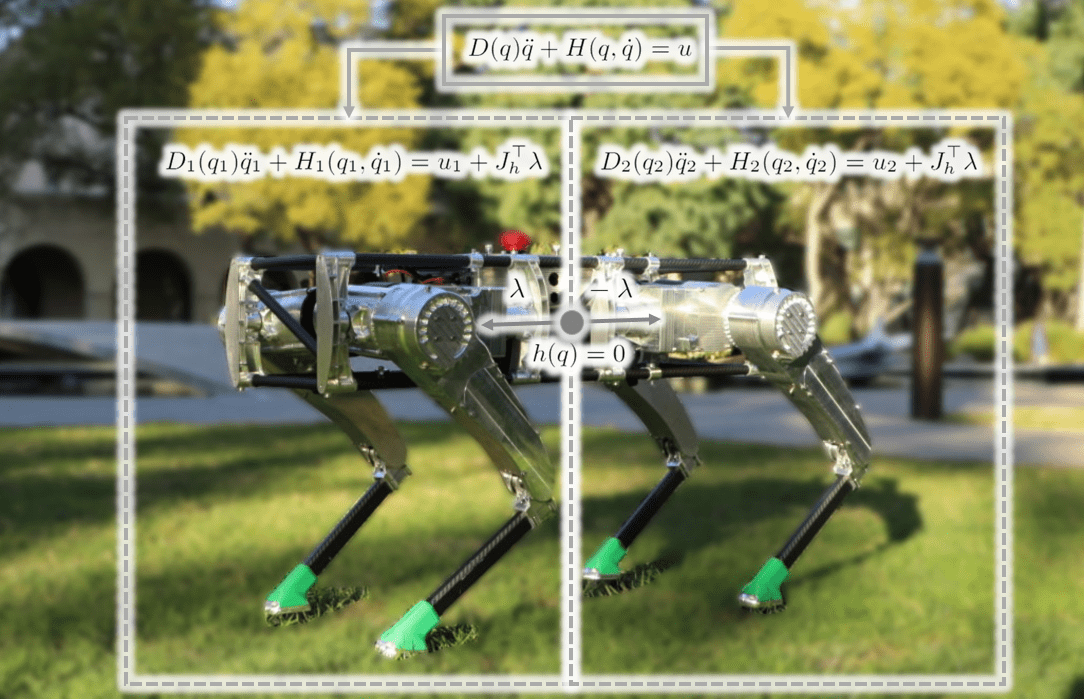}
	\vspace{-2.5mm}
	\caption{Conceptual illustration of the full body dynamics decomposition, where the 3D quadruped --- the Vision 60 --- is decomposed into two constrained 3D bipedal robots.}
	\label{fig:frontPage}
	\vspace{-4.8mm}
\end{figure}

This paper generalizes the aforementioned methods --- zero dynamics and system decomposition through coupling constraints --- and unifies them through a novel formulation: {\it coupled control systems}. We then utilize zero dynamics to reduce to a subsystem dependent on coupling constraint which is then eliminated via coupling relations to yield the final 
isolated subsystem. The main result of this paper is that solutions of the isolated subsystem are solutions of the full-order system, and thus periodic orbits on the subsystem yield periodic orbits on the full-order system. This result is leveraged to construct a nonlinear optimization problem utilizing collocation methods to generate these periodic solutions.

Our motivating application is gait (periodic orbit) generation for quadrupedal robots.
%
Previously, HZD methods were applied to quadrupedal walking \cite{ma2019First}; yet the high dimensionality of this system made it computationally expensive to generate gaits when compared to their bipedal analogs. To address this shortcoming, recent work has aimed at decomposing quadruped into bipedal robots \cite{ma2019bipedal}---it is this methodology that this paper formalizes and extends. 
Therefore, we consider a quadrupedal robot utilizing the \system~ paradigm, wherein this system can be reduced to lower-dimensional subsystems on which periodic orbits (gaits) can be generated. 
We demonstrate the results through the realization of these generated gaits experimentally to achieve stable walking on rough terrains.

\section{Coupled Control Systems} \label{sec:def}

This section introduces the notion of \systems, for which a collection of differential equations are coupled via algebraic coupling condition. The goal is to present the basic paradigm used throughout the paper. 

We first introduce a bidirectional graph $\Graph=(\V,\E)$ where the vertices $\V=\{1,2\}$ represent the indices of the subsystems and edges $\E=\{(1,2),(2,1)\}$ represent their connections. We then denote $\X=\{\X_i\}_{i\in\V}$ as a set of internal states, $\Z=\{\Z_i\}_{i\in\V}$ as a set of coupled states, and $\Us=\{\Us_i\}_{i\in\V}$ as a set of admissible control inputs. 
In addition, we  assume $i\neq j\in\Q$ and $e=(i,j),\overline{e}=(j,i)\in\E$ throughout the paper. 

We can now define the main object of interest. 
\vspace{-1.4mm}
\begin{definition} \label{def:systems}
    \emph{ A \defit{\system~(CCS)} $\CC$ is defined on a graph $\Graph$ 
    and a conditional expression:
    }
    \par\vspace{-3.3mm}{\small 
    \begin{align}\label{eq:symmetric_dynamics}
    \CC \triangleq &\begin{cases}
        \dot{x}_i = f_i(x_i,z_i) + g_i(x_i,z_i) u_i + \breve{g}_{e}(x_i,z_i,z_j) \lambda_{e} \\ 
        \dot{z}_i = p_i(x_i,z_i) + q_i(x_i,z_i) u_i + \breve{q}_{e}(x_i,z_i,z_j) \lambda_{e}  \\
        \mathrm{s.t.} 
        \quad  c_{e}(z_i, z_j) = z_i - z_j =-c_{\bar{e}}(z_j,z_i)\equiv 0 \\
        \hspace{0.73cm} \lambda_e = - \lambda_{\overline{e}}, 
    \end{cases} 
    \end{align}}%
    \vspace{-3mm}\par\noindent
    \emph{where,
    $x_i\in\X_i, z_i\in\Z_i, u_i\in\Us_i$, and $c_e(z_i,z_j) \equiv 0 $ is a \emph{coupling constraint} enforced by the \emph{coupling inputs} $\lambda_e$, where $\equiv$ represents the identical equality of functions.
    }
\end{definition}
\vspace{-1mm}
We additionally denote 
$x = \{x_1, x_2\}\in\X$, 
$z = \{z_1, z_2\}\in\Z$,
$u = \{u_1, u_2\}$ and
$\lambda = \{\lambda_e, \lambda_{\bar e}\}$ throughout the paper.

\newsec{Solutions.}  
We define solutions of \systems\ by assuming the existence of feedback control laws: $u(x,z) \defeq \{u_1(x_1,z), u_2(x_2,z)\}$. 
Applying these controllers to \eqref{eq:symmetric_dynamics} yields a \defit{coupled dynamical system (CDS)}: 
\par\vspace{-4mm}{\small
\begin{align} \label{eq:controlled_dynamicsc}
\CD \triangleq &\begin{cases}
    \dot{x}_i = f_i^{\cl}(x_i,z) + \breve{g}_{e}(x_i, z) \lambda_{e} \\ 
    \dot{z}_i = 
    p_i^{\cl}(x_i,z) + \breve{q}_{e}(x_i,z) \lambda_{e}  \\
    \mathrm{s.t.} \quad  c_{e}(z) \equiv 0, 
    \quad \lambda_e = - \lambda_{\overline{e}} 
    \end{cases} 
\end{align}}%
\vspace{-3mm}\par\noindent
where, $f_i^{\cl} \defeq f_i(x_i,z_i) + g_{i}(x_i,z_i) u_i(x_i,z)$, 
and 
$p_i^{\cl} \defeq p_i(x_i,z_i) + q_{i}(x_i,z_i) u_i(x_i,z).$ 
Then the \defit{solution} of the coupled dynamic system, $\CD$, is a set of solutions: 
\par\vspace{-4mm}{\small 
\begin{align*} 
    \Big\{ \big( x_1(t),z_1(t),\lambda_e(t) \big), \big( x_2(t), z_2(t),  \lambda_{\bar{e}}(t) \big) \Big\} \ 
    \mathrm{s.t.}\  \eqref{eq:controlled_dynamicsc} \ \forall t\in\mathbf{I}\subset\R 
\end{align*}}%
\vspace{-4mm}\par\noindent 
with initial condition: 
$\big\{ ( x_1(0),z_1(0),\lambda_e(0) ), ( x_2(0), z_2(0), \\ 
\lambda_{\bar{e}}(0) ) \big\}$, and $\mathbf{I}\subset\R$ is the time interval of their existence. Per the above notation, we will sometimes denote the solutions by $(x(t),z(t),\lambda(t))$ with initial condition $(x(0),z(0),\lambda(0))$.

\begin{figure}[t]
	\centering
	\vspace{2mm}
	\includegraphics[width=0.45\textwidth]{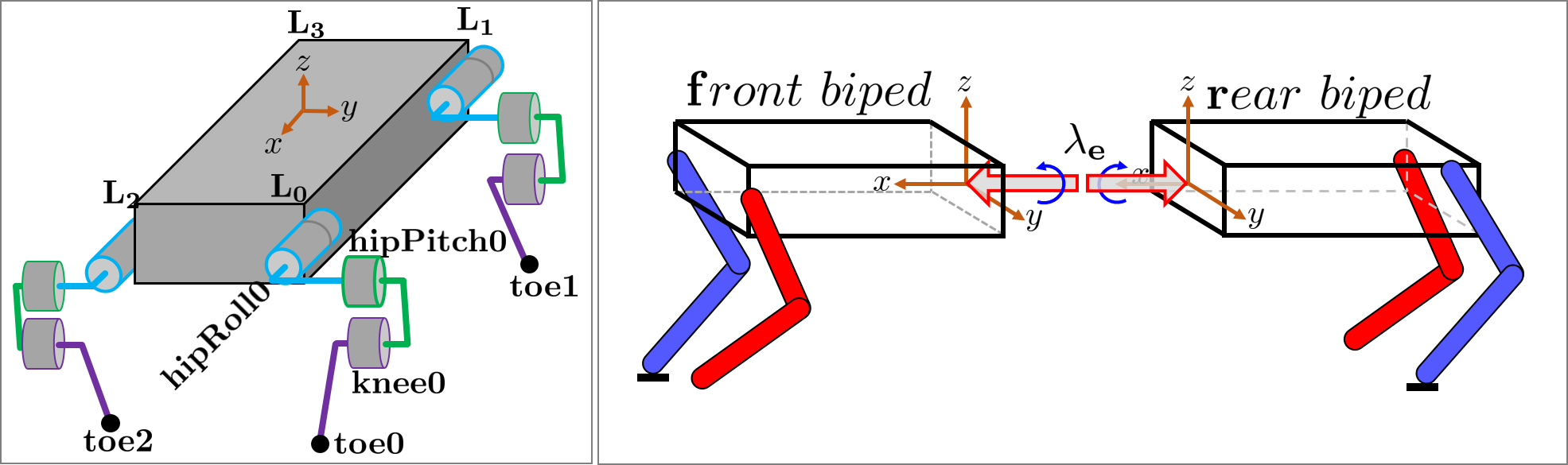}
	\vspace{-3mm}
	\caption{{\small
	Left: the configuration coordinates of the quadruped, 
	each leg of which has a point contact toe. 
	Right: the decomposition of a quadrupedal robots into two bipedal systems. 
	}}\label{fig:quad}
	\vspace{-10mm}
\end{figure}

\newsec{Coupling constraints.} 
Importantly, the solutions must satisfy the coupling constraints 
at all time. Therefore, 
\par\vspace{-4mm}{\small 
\begin{align}
    \label{eq:cdotcond}
    &c_e(z) \equiv 0 \ \Rightarrow\ \dot{c}_e(z, \dot z) \equiv 0  
    \\
    &\hspace{14mm}
    \Rightarrow \underbrace{\frac{\partial c_e(z_i,z_j) }{\partial z_i}}_{\defeq \Je{(i,j)}(z)} \dot{z}_i + 
    \underbrace{\frac{\partial c_e(z_i,z_j) }{\partial z_j}}_{\defeq  \Je{(j,i)}(z)}\dot{z}_j 
    \equiv 0  \notag
    \\
    \label{eq:cdotconstraint}
    &\Rightarrow
    \dot{c}_e(x,z)
    = \Je{(i,j)}(z) \left( p_i^{\cl}(x_i,z) 
            + \breve{q}_{e}(x_i, z) \lambda_{e}  
        \right) \notag \\
    &\hspace{13mm}+ 
        \Je{(j,i)}(z)\left( p_j^{\cl}(x_j,z) 
            + \breve{q}_{\bar{e}}(x_j, z) \lambda_{\bar{e}} 
        \right) 
    \equiv 0
\end{align}}%
\vspace{-5mm}\par\noindent 
Hence, to solve for the coupling inputs $\lambda_e$ that satisfy the coupling constraints, it is necessary to solve an equation that depends on the states of both subsystems. To address this, we present a method for isolating a subsystem via conditions on the controllers of the other systems in the next section. 
%
%
Before doing this, we utilize the following example to illustrate the concepts of \systems.

\newsec{Application to quadrupedal robots.}
The motivating application considered here, is to compute periodic solutions of the quadrupedal dynamics. As \figref{fig:quad} shown, we decompose this quadruped into two bipeds, whose dynamics are on a CCS graph (according to definition \ref{def:systems}): 
$\Graph \defeq \big( \V = \{\rmf,\rmr\}, \E=\{e=(\rmf,\rmr),\bar e=(\rmr,\rmf)\} \big)$, 
where $\rmf,\rmr$ label the \textit{front} and \textit{rear} bipedal systems, correspondingly. We picked the coordinates for these two subsystems as 
$   
    q_\rmf = (\xi_\rmf^\top, \theta_{\rmL_2}^\top, \theta_{\rmL_0}^\top)^\top, 
    q_\rmr = (\xi_\rmr^\top, \theta_{\rmL_1}^\top, \theta_{\rmL_3}^\top)^\top
$ 
with $\xi_i\in\R^3\times\mathrm{SO}(3)$ and the leg joints $\theta_{L_*}\in\R^3$. Since all leg joints are actuated, the inputs are $u_i\in\mathcal{U}\subset\R^6$.
The decomposed dynamics of a quadruped as two coupled bipeds in the \textit{continuous phase}\footnote{
The definition for continuous and discrete dynamics (impact dynamics) for hybrid control systems can be found in \cite{Grizzle2014Models}, which is less relevant to the main theme here, hence omitted.
}, are given by a set of DAEs: 
\par\vspace{-4mm}{\small 
\begin{numcases}{\quadruped\defeq}
    D_i \ddot q_i + H_i = J_i^\top F_i + B_i u_i + J_e^\top\lambda_e
    \label{eq:ol21}\\
    J_i \ddot q_i + \dot J_i \dot q_i = 0
    \label{eq:ol22}\\
    \mathrm{s.t.\quad} 
    c_e(\xi_i, \xi_j) = \xi_i - \xi_j \equiv 0
    \label{eq:ol23} \\
    \hspace{8mm} \lambda_e=-\lambda_{\bar{e}} 
    \label{eq:ol24} 
\end{numcases}}%
\vspace{-3mm}\par\noindent 
with $D_i(q_i)\in\R^{n\times n}$ the mass-inertia matrix, 
$H_i(q_i, \dot q_i)\in\R^n$ the drift vector,  
and the actuation matrix $B_i=\begin{bmatrix}\0_{6\times 6} & I_{6\times 6} \end{bmatrix}$. 
The contact (holonomic) constraint $h_i(q_i)\equiv 0$ is enforced via ground reaction forces $F_i\in\R^3$, whose second derivative is given in \eqref{eq:ol22}. More details of these notations can be found in \cite{ma2019bipedal}.
%
Note that $F_i$ can be eliminated by the solving \eqref{eq:ol21}-\eqref{eq:ol22} to have a shorter form: $D \ddot q + \bar H = \bar B u + \bar J_e^\top \lambda_\rmc$. 
The derivation is straightforward hence ignored.

To obtain a CCS as in \eqref{eq:symmetric_dynamics}, we pick ``normal form'' type coordinates  (see \cite{Sastry1999Nonlinear}), with the ``output'' 
(also known as virtual constraint \cite{Westervelt2007a}) 
that we wish to zero, given by 
\par\vspace{-4.5mm}{\small 
\begin{align} 
    y_i(q_i, \alpha_i) = y^a(q_i) - y^d( \xi_i, \alpha_i ), 
    \label{eq:quadoutput}
\end{align}}%
\vspace{-5mm}\par\noindent 
where $y^a, y^d$ are the actual and desired outputs, $\xi_i$ represents a parameterization of time and $\alpha_i\in\R^{6\times 6}$ are the coefficients for six $5^{\mathrm{th}}$-order Be\'zier polynomials that are designed by the optimization algorithm in Sec.\ref{sec:opt}. Since our goal is to find a \textit{symmetric ambling} gait for quadrupeds, we chose $\alpha_\rmr = \mathcal{M} \alpha_\rmf $, with the matrix $\mathcal{M}$ representing a mirroring relation. 
It is important to note that the output coordinate here utilizes a state-feedback structure, instead of the time-based construction of \cite{ma2019bipedal}. 
We can then construct our internal states $x_i = (y_i^\top, \dot{y}_i^\top)^\top$, leaving the coupled states as $z_i = (\xi_i^\top, \dot{\xi}_i^\top)^\top$.  
The end result is a CCS of the form given in \eqref{eq:symmetric_dynamics} for this mechanical system: 
\par\vspace{-3.8mm}{\small 
\begin{align*} 
    &\dot x_i =
        \underbrace{\begin{bmatrix} \dot{y}_i \\ \dot J_{y_i}\dot{q}_i -J_{y_i} D_i^{-1}\bar{H}_i \end{bmatrix}}_{f_i(x_i,z_i)} + 
        \underbrace{\begin{bmatrix} 0 \\ J_{y_i} D_i^{-1}\bar{B}_i \end{bmatrix}}_{g_i(x_i,z_i)}u_i + 
        \underbrace{\begin{bmatrix} 0 \\ J_{y_i} D_i^{-1}\bar{J}_e^\top \end{bmatrix}}_{\breve{g}_{e}(x_i,z_i,z_j)} \lambda_{e} 
    \vspace{1mm}\\ 
    &\dot z_i = 
        \underbrace{\begin{bmatrix} \dot\xi_i \\ -J_\xi D_i^{-1}\bar{H}_i \end{bmatrix}}_{p_i(x_i,z_i)} + 
        \underbrace{\begin{bmatrix} 0 \\ J_\xi D_i^{-1} \bar{B}_i \end{bmatrix}}_{q_i(x_i,z_i)} u_i + 
        \underbrace{\begin{bmatrix} 0 \\ J_\xi D_i^{-1} \bar{J}_e^\top  \end{bmatrix}}_{\breve{q}_{e}(x_i,z_i,z_j)} \lambda_{e}
    \vspace{1mm}\\ 
    &\ \ 
    \mathrm{s.t.\quad} c_e(z_i, z_j) = z_i - z_j
\end{align*}}%
\vspace{-5.5mm}\par\noindent 
with $J_{y_i} = \partial y_i(q_i)/\partial q_i$, $J_\xi=\partial \xi/\partial q = \begin{bmatrix}I_{6\times 6} & 0_{6\times 6}\end{bmatrix}$, where we suppressed the dependency on $x_i, z_i$ for all entries.

\section{Isolating Control Subsystems} \label{sec:isolating}

The main idea in approaching the analysis and design of controllers for \systems\ is to isolate subsystems that encode the behavior of the overall CCS. This section outlines the procedure for isolating the subsystems through a two-step approach: restricting systems to the \textit{zero dynamics} manifold, and leveraging this to explicitly calculate the coupling conditions. We then can reduce the full-order CCS to a subsystem that no longer depends on the internal states of the other subsystem. We establish the main result of the paper encapsulating these constructions: solutions of the subsystem yield solutions of the full-order dynamics.  


\subsection{$\lambda$-Coupled Subsystem}

Given a CCS $\CC$, we define the \defit{zero dynamics manifold} for each subsystem $i\in\Q$ as: 
\par\vspace{-4mm}{\small 
\begin{align}\label{eq:zerodynman}
    \Zbf_i \defeq \{ (x,z) \in \X \times \Z ~ | ~ x_i \equiv 0 \}.
\end{align}}%
\vspace{-5mm}\par\noindent 
Thus, the zero dynamics manifold for \ith subsystem consists of the internal states, $x_i$, being zero, i.e., the system evolves only according to the coupled states $z$. 

We wish to design controllers of the overall CCS on the zero dynamics of subsystem $j\in\Q$. Therefore, a controller $\uZlambda_j(x_j,z)$ is said to \defit{render the zero dynamics manifold $\Zbf_j$ invariant} if it satisfies the following algebraic condition:
\par\vspace{-4mm}{\small
\begin{align}\label{eq:zerodyncond}
    0 \equiv f_j(0,z_j) + g_j(0,z_j) \uZlambda_{j}(0,z) + \breve{g}_{\overline{e}}(0,z) \lambda_{\bar e}
\end{align}}%
\vspace{-5mm}\par\noindent 
where $\uZlambda_j$ implicitly depends on $\lambda_{\bar e}$ for $\bar e=(j,i)\in\E$. 
By applying $\uZlambda_j$, we obtain a \defit{$\lambda$-coupled control subsystem ($\lambda$-CCSub)} for the \ith subsystem:
\par\vspace{-4mm}{\small
\begin{align}\label{eq:lambdasubsystemsym}
    \SCClambda{i} \defeq 
    \begin{cases}
        \dot{x}_i = f_i(x_i,z_i) + g_i(x_i,z_i) u_i +  \breve{g}_{e}(x_i,z) \lambda_{e} \\ 
        \dot{z}_i =  p_i(x_i,z_i) + q_i(x_i,z_i) u_i + \breve{q}_{e}(x_i,z) \lambda_{e} \\
        \dot{z}_j = p_j(0,z_j) + q_j(0,z_j) \uZlambda_j(0,z) + \breve{q}_{\bar{e}}(0,z) \lambda_{\bar e}  \\
        \mathrm{s.t.} \quad c_e(z) = z_i - z_j \equiv 0, 
        \quad \lambda_e = - \lambda_{\bar{e}}
        \raisetag{10pt}
    \end{cases} 
\end{align}}%
\vspace{-3mm}\par\noindent 
Thus, the \ith subsystem evolves according to its own dynamics and the zero dynamics of all remaining systems---all of which are coupled via the coupling inputs $\lambda$. 

\subsection{Explicit Coupling Conditions}
The coupling between the control systems \eqref{eq:symmetric_dynamics} is enforced via $\lambda$ and the coupling constraints of the form \eqref{eq:cdotconstraint}.  Similarly, even in the reduction to a subsystem \eqref{eq:lambdasubsystemsym}, the coupling is still achieved through $\lambda$. We wish to generalize this so as to remove the coupling, i.e., isolate subsystems, while still preserving the overall behavior of the full system. 
We first define the \defit{coupling relation} that allows the use of the controllers $\uZlambda_j$ to eliminate the dependence on the controllers and internal states of the other subsystem. 


\vspace{-1mm}
\begin{definition}\label{def:couplingrelation}
    \emph{For a $\lambda$-CCSub $\SClambda$ and $i\in\Q$, a \defit{coupling relation} is a functional relationship on the coupling inputs
    }
    \par\vspace{-4mm}{\small 
    \begin{align}\label{eq:lambdaeform}
        \lambdaZ_{e}(x_i,z;u_i) =  \AZ_{e}(x_i,z) u_i +  \bZ_{e}(x_i,z), 
    \end{align}}
    \vspace{-5mm}\par\noindent 
    \emph{that satisfies the coupling constraint \eqref{eq:cdotcond} for all $e=(i,j)\in\E$. }
\end{definition}


The coupling relation is then summarized in the following:
\vspace{-5mm}
\begin{lemma}\label{lemma:coupling}
    For a CCS $\CC$, if we have 
    \par\vspace{-3mm}{\small 
    \begin{align*}
        \breve{Q}_{e}(x_i,z)  \defeq 
        \begin{bmatrix}
        g_j(0,z_j) & \breve{g}_{\overline{e}}(0, z) \\
        q_j(0,z_j) & \breve{q}_{e}(x_i, z) + \breve{q}_{\overline{e}}(0, z)
        \end{bmatrix}
    \end{align*}}%
    \vspace{-4mm}\par\noindent 
    invertible, there exists a controller $\uZ_j$ that renders $\Zbf_j$ invariant and a coupling relation in \eqref{eq:lambdaeform}, given by: 
    \par\vspace{-4mm}{\small
    \begin{align}
        \begin{bmatrix}
            \uZ_j(0,z; u_i) \vspace{0.5mm}\\
            \lambdaZ_{e}(x_i,z; u_i)
        \end{bmatrix}  =  
        \breve{Q}_e^{-1} 
        \left(
            \begin{bmatrix} 0 \\ q_i(x_i,z_i) \end{bmatrix} u_i
            \hspace{-1mm} + \hspace{-1mm}
            \begin{bmatrix}
                -f_j(0,z_j) \\
                p_i(x_i,z_i) - p_j(0,z_j)
            \end{bmatrix} 
        \right) \notag
    \end{align}}\noindent
\end{lemma}%
\hspace{-12mm}
\begin{proof}
    Evaluating \eqref{eq:cdotcond} along the zero dynamics manifold $\Zbf_j$, i.e., $x_j \equiv 0$, yields: 
    $
    q_i(x_i,z_i)u_i + p_i(x_i,z_i) - p_j(0,z_j) = 
    q_j(0,z_j)\uZ_j(0,z;u_i) -  \left(\breve{q}_{e}(x_i,z) + \breve{q}_{\bar{e}}(0,z) \right) \lambda_e.  
    $
    Combining this with \eqref{eq:zerodyncond} and simultaneously solving for $\uZ_j$ and $\lambdaZ_{e}$ yields the desired result. 
\end{proof}

Recall that the controller $\uZlambda_j$ that renders the zero dynamics surface invariant implicitly depends on $\lambda_{\overline{e}}$ via \eqref{eq:zerodyncond}. Now with a coupling relation, the dependence of $\lambda_{\overline{e}}$ is removed, and as a result we say that $\uZ_j$ \textit{renders the zero dynamics manifold $\Zbf_j$ invariant} if: 
\par\vspace{-4mm}{\small
\begin{align}\label{eq:zerodynamicswithlambdarelation}
    0 \equiv \fZj(0,z) + \gZj(0,z) u_i + g_j(0,z_j) \left( \uZ_j(0,z;u_i)  - u_i \right)
\end{align}}%
\vspace{-5mm}\par\noindent 
where $\uZ_j$ is now a function of $u_i$ and 
\par\vspace{-4mm}{\small 
\begin{align}
    \begin{cases} \label{eq:fZjgZj}
        \fZj(x_j,z) & \defeq f_j(x_j,z_j) - 
        \breve{g}_{\bar{e}}(x_j,z)\bZ_{e}(x_i,z), 
        \\
        \gZj(x_j,z) & \defeq g_j(x_j,z_j) - 
        \breve{g}_{\bar{e}}(x_j,z)\AZ_{e}(x_i,z).
    \end{cases}
\end{align}} %
\par\vspace{-3.5mm}


Returning to \eqref{eq:cdotconstraint}, given a coupling relation we can rewrite this coupling constraint as: 
\par\vspace{-4mm}{\small 
\begin{align} \label{eq:cdotconstraintlambda}
    \hspace{-6mm}
    \dot{c}_e(x_i,z)   
    &= \Je{(i,j)}(z) 
        \left( \pZi(x_i,z) +  \qZi(x_i,z) u_i \right)
    \notag\\
    &\hspace{1mm}
    + \Je{(j,i)}(z)
        \left( \pZj(x_i,z) + \qZj(x_i,z) u_i \right) \equiv 0 
\end{align}}%
\vspace{-5mm}\par\noindent
where for the subsystem $\SClambda$ we have 
\par\vspace{-4mm}{\small 
\begin{align} \label{eq:pqZ}
    \begin{cases}
        \pZi(x_i,z) &\defeq p_i(x_i,z_i) + 
        \breve{q}_{e}(x_i,z_i,z_j)\bZ_e(x_i,z) 
        \\
        \qZi(x_i,z) &\defeq q_i(x_i,z_i) +  
        \breve{q}_{e}(x_i,z_i,z_j)\AZ_e(x_i,z) 
        \\
        \pZj(x_i,z) &\defeq p_j(0,z_j)  + q_{j}(0,z_j) \uZ_j(0,z)  
        \\
        & \hspace{4mm} 
        - \breve{q}_{\bar{e}}(0,z_j,z_i) \bZ_{e}(x_i,z) 
        \\
        \qZj(x_i,z) &\defeq - 
        \breve{q}_{\bar{e}}(0, z_j, z_i) \AZ_{e}(x_i, z) 
    \end{cases}
\end{align}}%

\subsection{Isolating Subsystems} \label{sec:isolatingsubsystems}
We now arrive at the key concept for which all of the previous constructions have built --- reducing a CCS to a single subsystem that can be used to give guarantees about the entire CCS. 
This is based on the following definition.
\vspace{-1mm}
\begin{definition}\label{def:isolatedsubsystem}
    \emph{
    For a CCS $\CC$, and $i\neq j \in \Q$, assume a coupling relation $\lambdaZ_{e}$ such that there exist $\uZ_j$ rendering the zero dynamics manifold $\Zbf_j$ invariant. Then the \ith \defit{control subsystem (CSub)} associated with the CCS $\CC$ is given by: 
    }
    \par\vspace{-4mm}{\small
    \begin{align}\label{eq:subsystem1}
        \SCZ \defeq 
        \begin{dcases}
            \dot{x}_i  =  \fZi(x_i,z) + \gZi(x_i,z) u_i  \\ 
            \dot{z}_i  =  \pZi(x_i,z) + \qZi(x_i,z) u_i  \\
            \dot{z}_j  =  \pZj(x_i,z) + \qZj(x_i,z) u_i \quad 
        \end{dcases} 
    \end{align}
    }%
    \vspace{-2mm}\par\noindent
   \emph{where 
    $ \fZi(x_i,z) \defeq f_i(x_i,z_i) + \breve{g}_{e}(x_i,z_i,z_j)\bZ_e(x_i,z) $, 
    $ \gZi(x_i,z) \defeq g_i(x_i,z_i) + \breve{g}_{e}(x_i,z_i,z_j)\AZ_e(x_i,z) $, 
    and $\pZi, \qZi,$ $\pZj, \qZj$ are given in \eqref{eq:pqZ}. 
    Furthermore, when a feedback controller $u_i(x_i,z)$ is applied to $\SCZ$, the result is a dynamical system, denoted by $\SDZ$.}
\end{definition}
\vspace{-1mm}

Note that the coupling constraint \eqref{eq:cdotconstraintlambda} was not explicitly stated in the CSub $\SCZ$. This was because it was solved for via the coupling relation $\lambdaZ_e$. That is, the system naturally evolves on the \defit{constraint manifold}: 
$
    \Cbf \defeq \{ (x,z) \in  \X \times \Z : c_e(z) \equiv 0, ~ \forall ~ e \in \E \}. 
$
This is made formal in the following result. 
Additionally, it will be seen that solutions of the \ith subsystem, denoted by $(x_i(t),z(t),\lambda(t))$, can be used to construct solutions of the full-order CCS.  
%
%
Before formally stating the ultimate result of this paper, we need some notation. Let $(x_i, z)\in\X_i\times\Z$ and consider the canonical embedding $\iota: \X_i\times\Z \hookrightarrow \X\times\Z$ given by $\iota(x_i,z) = (x,z) $, 
where $x = \{x_i, x_j\}$ and $x_j = 0$.

\begin{theorem}\label{thm:subsystem}
    Let $\CC$ be a CCS, and for the \jth system assume there exist $\uZ_j$ 
    that render the zero dynamics manifold $\Zbf_j$ invariant. Let $\SCZ$ be the corresponding $\lambda$-CCSub for \ith subsystem. Given a feedback controller $u_i(x_i,z)$ for the CSub with corresponding dynamical subsystem $\SDZ$ with solution $(x_i(t),z(t))$ for $t \in I \subset \R$.  If 
    \par\vspace{-4mm}{\small
    \begin{align*}
        \iota(x_i(0),z(0)) \in \Cbf \quad \Rightarrow \quad \iota(x_i(t),z(t)) \in \Cbf \quad \forall ~ t \in I \subset \R
    \end{align*}}%
    \vspace{-5mm}\par\noindent 
    and $(\iota(x_i(t),z(t)),\lambdaZ(t))$,  with 
    \par\vspace{-4mm}{\small
    \begin{align*}
        \lambdaZ(t) = \Big\{\lambdaZ_e\big( 
        x_i(t),z(t);u_i(x_i(t),z(t))
        \big)\Big\}_{e \in \E}
    \end{align*}}%
    \vspace{-5mm}\par\noindent 
    is a solution of $\CD$, the CDS obtained by applying $u_i, \uZ_j$.
\end{theorem}

\begin{proof}
    The condition that $(x(0),z(0)) \in \Cbf$ is equivalent to $c_e(z(0)) = 0$. Concretely, $c_e(z_i(0),z_j(0)) = 0$. Since $\lambdaZ_e$ is a coupling relation, it satisfies \eqref{eq:cdotconstraint} and more explicitly \eqref{eq:cdotconstraintlambda}; therefore, and being explicit about the arguments, $\dot{c}_e(x(t),z(t)) = 0$ for all $t \in \mathbf{I}$ and all $e \in \E$.  It follows that $c_e(z(t)) = 0$ for all $t \in \mathbf{I}$ and $e \in \E$. 
    
    The fact that $\big(\iota(x_i(t),z(t)),\lambdaZ(t)\big)$ is a solution of $\CD$ assuming that $(x_i(t),z(t))$ is a solution of $\SDZ$ 
    %
    follows trivially from the fact that the zero dynamics $\Zbf_j$ are invariant, i.e., 
    $
    \iota\big( x_i(t),z(t) \big) \in \Zbf_{j}, 
    \ \forall ~ t \in \mathbf{I}. 
    $
\end{proof}

\newsec{Periodic Orbits.}  In the context of quadrupedal dynamics, we will be interested in generating periodic solutions, i.e., walking. A solution of a CDS $\CD$ is \defit{periodic} of period $T > 0$ if for some initial condition $(x(0),z(0),\lambda(0))$:
\par\vspace{-4mm}{\small 
\begin{align*}
    \big( x(t+T),z(t+T),\lambda(t+T) \big) = (x(t),z(t),\lambda(t))
\end{align*}}%
\vspace{-5mm}\par\noindent 
with the resulting periodic orbit: 
$
\mathcal{O} = \{(x(t),z(t)) \in \X \times \Z ~ | ~ 0 \leq t \leq T \}. 
$
As a result of Theorem \ref{thm:subsystem}, periodic orbits in a subsystem correspond to the periodic orbits in the full-order dynamics.

\begin{corollary}\label{cor:periodic}
    Under the conditions of Theorem \ref{thm:subsystem}, assume that $(x_i(t),z(t))$ is a periodic solution of $\SDZ$ with period $T > 0$ and corresponding orbit $\mathcal{O}_i = \{ (x_i(t),z(t)) \in \X_i \times \Z ~ | ~ 0 \leq t \leq T\}$. Then $(\iota(x(t),z(t),\lambdaZ(t))$ is a periodic solution of the CDS with period $T > 0$ and corresponding periodic orbit $\mathcal{O} = \iota(\mathcal{O}_i)$. 
\end{corollary}

\vspace{-2mm}
\newsec{Application to quadrupeds. } 
For the quadrupedal dynamics $\quadruped$, since the output \eqref{eq:quadoutput} has (vector) relative degree $2$ with respect to $u_i$ (see \cite{Westervelt2007a}), we can explicitly design the controller $\uZlambda_j$ that renders $\Zbf_j$ invariant: 
\par\vspace{-4mm}{\small 
\begin{align*} 
    \uZlambda_j
    = (J_{y_i} D_j^{-1}\bar{B}_j)^{-1}\big( 
                J_{y_j} D_j^{-1}\bar{H}_j - \dot{J}_{y_j}\dot{q}_j
                - J_{y_i} D_j^{-1} \bar{J}_e^\top\lambda_e \big),
\end{align*}}%
\vspace{-5mm}\par\noindent
as given by Lemma \ref{lemma:coupling}. Hence, this controller satisfies \eqref{eq:zerodyncond} and renders a $\lambda$-coupled CSub,  as in \eqref{eq:lambdasubsystemsym}. 

For robotic systems, we take these ideas one step further to obtain ``bipeds'' that are the isolated subsystems associated with quadrupeds and include slack variables that are beneficial for gait generation. Operating on the invariant zero dynamics manifold $\mathbf{Z}_j$ yields 
$y_j(q_j, \alpha_j) \equiv 0$, hence
\par\vspace{-4mm}{\small 
\begin{align*}
\theta_a \equiv H_a^{-1}y^d(\xi_j,\alpha_j) 
    &~\text{and}~
    q_j^\Zbf(\xi_j) \equiv \big( \xi_j^\top, ( H_a^{-1} y^d(\xi_j, \alpha_j) )^\top \big)^\top
    \\
    &\hspace{-22mm}
    \Rightarrow\ \quad \ddot q_j^\Zbf(\xi_j, \dot\xi_j, \ddot\xi_j) 
        = J_\rmz(\xi_j) \ddot \xi_j + \dot J_\rmz(\xi_j,\dot \xi_j) \dot \xi_j.
\end{align*}}%
\vspace{-5mm}\par\noindent 
where $J_\rmz = \partial q_j^\Zbf(\xi_j) / \partial\xi_j$. In another word, if $\uZlambda_j$ exists and is applied to \jth subsystem, the $j^{\mathrm{th}}$ bipedal dynamics given by in \eqref{eq:ol21}-\eqref{eq:ol22} are equivalent to:
\par\vspace{-4mm}{\small 
\begin{numcases}{ }
    D_j \ddot{q}_j^\Zbf(\xi_j, \dot\xi_j, \ddot\xi_j) + H_j = J_j^\top F_j + B_j u_j^\Zbf + J_e^\top \lambda_e
    \label{eq:bi1}\\
    J_j \ddot{q}_j^\Zbf(\xi_j, \dot\xi_j, \ddot\xi_j) + \dot{J}_j \dot{q}_j^\Zbf(\xi_j, \dot\xi_j) = 0
    \label{eq:bi2}
\end{numcases}}\noindent
\vspace{-3mm}\par\noindent 
where for simplicity we have suppressed the dependencies of  $D_j(q_j(\xi_j)), J_j(q_j(\xi_j))$ and $H_j(q_j(\xi_j), \dot{q}_j(\xi_j,\dot{\xi}_j))$. 
We then leverage a specific structure of rigid-body dynamics when using the floating base convention: 
$ B_j u_j + J_{\bar{e}}^\top \lambda_e = (\lambda_{\bar{e}}^\top, u_j^\top)^\top$. 
Utilizing this, \eqref{eq:bi2} and the first 6 rows of \eqref{eq:bi1} yield the following ``bipedal'' dynamics:
\par\vspace{-4mm}{\small 
\begin{align}
    \bipedz{j} \defeq
    \begin{cases}
        D^\Zbf_j \ddot\xi_j + H^\Zbf_j = \hat{J}_j^\top F_j + \lambda_e\\
        J_j^\Zbf \ddot\xi_j + w_j^\Zbf = 0 
    \end{cases}
    \label{eq:bizieom}
\end{align}}%
\vspace{-3mm}\par\noindent 
with 
$
    D^\Zbf_j  = \hat{D}_j J_\rmz, 
    H^\Zbf_j  = \hat{D}_j \dot{J}_\rmz\dot\xi_j + \hat{H}_j, 
    J_j^\Zbf  = J_j J_\rmz,  
$
and
$  w_j^\Zbf  = J_j\dot{J}_\rmz\dot\xi_j + \dot{J}_j \dot{J}_\rmz \dot{\xi}_i $. 
Here,  
$\hat{\square}$ are the first 6 rows (block) of the variable $\square$. 
Hence, $\bipedzj$ represents the dynamics of a subsystem $j$ on $\Zbf_j$, i.e., \eqref{eq:bizieom} evolves according to \eqref{eq:zerodyncond} by adding a slack variable $F_j$ that can be uniquely determined.

\section{Coupled System Optimization} \label{sec:opt}

\newcommand{\id}{{\kappa}} 
\newcommand{\vi}{\vartheta^\id} 
\newcommand{\xv}{\chi} 

With the previous construction of \systems, we present a general optimization framework to solve for the solution of the \ith CSub in \eqref{eq:subsystem1} associated with the CCS, while synthesising the controllers that render forward invariance of the zero dynamics manifolds. 
%
%
%
The approach we will take is a \textit{locally direct collocation} based optimization method \cite{hereid_dynamic_TRO}, which has been widely applied to finding solutions to dynamical systems such as 
\cite{reher2016realizing}. 
We now pose the previous formulations as a series of constraints to represent the controlled dynamics of $\SCZ$. 
Along this process, the problem formulation of our target application --- the control of quadrupedal walking, will be used as an example to illustrate this method.

\newsec{Optimization setup.} 
We first discretized the time horizon $t\in[0,T]$ evenly to obtain the grid indices $\id = 0,1,...\rmK$, i.e., $t^\id = T \id/\rmK$. 
We define the \textit{decision variable} associated with the \ith control subsystem $\SCZ$ as: 
\par\vspace{-4mm}{\small 
\begin{align*}
    \mathbf{X} \defeq \big\{ \vartheta^\id \big\}_{\id = 0,1,...K},
    \quad 
    \vartheta^\id \defeq \{x_i^\id, \dot x_i^\id, z_i^\id, \dot z_i^\id,
    z_j^\id, \dot z_j^\id, u_i^\id, u^{\Zbf,\id}_j \}
\end{align*}}%
\vspace{-6mm}\par\noindent 
Note that we abbreviated the dependency on time $t$ as $\square^\id \defeq \square(t^\id)$ for notational simplicity. 

Recall that given a coupling relation, we have associated zero dynamics invariance conditions given by \eqref{eq:zerodynamicswithlambdarelation}. We will enforce these conditions in the optimization to ensure that $u^{\Zbf,\id}_j$ renders $\Zbf_j$ invariant as:
\par\vspace{-4mm}{\small
\begin{align*} 
    & F_{\mathrm{zero}}(\vartheta^\id)  \defeq  
    \fZj(0,z^\id) + \gZj(0,z^\id) u_i^\id + g_j(0,z_j^\id) 
    \left( u^{\Zbf,\id}_j - u_i^\id \right),
\end{align*}}%
\vspace{-5mm}\par\noindent 
where $\fZj$ and $\gZj$ are given as in \eqref{eq:fZjgZj}.


Next, following from the constructions in Sec.\ref{sec:isolatingsubsystems}, we define constraints corresponding to the dynamics of the \ith control subsystem $\SCZ$ (as obtained from the coupling relation). Denote 
$\chi^\id = (x_i^\id, z_i^\id, z_j^\id)$ and
\par\vspace{-4mm}{\small 
\begin{align*} 
    F(\xv^\id, u_i^\id) & \defeq   
    \begin{cases}
         \fZi(x_i^\id,z^\id) + \gZi(x_i^\id,z^\id) u_i^\id  \\ 
         \pZi(x_i^\id,z^\id) + \qZi(x_i^\id,z^\id) u_i^\id  \\
         \pZj(x_i^\id,z^\id) + \qZj(x_i,z) u_i^\id     
    \end{cases}
\end{align*}}%
\vspace{-3mm}\par\noindent
to obtain the \textit{dynamic constraints} as
\par\vspace{-4mm}{\small 
\begin{align} \label{eq:subsystemdyncons}
    F_{\mathrm{dyn}}(\vartheta^\id) \defeq \dot\xv^\id - F(\xv^\id, u_i^\id) = 0, 
    \tag{C.2}
\end{align}}%
\vspace{-5mm}\par\noindent 
which is an equality constraint imposed on the $\id^{\mathrm{th}}$ node to enforce all of the states and controllers satisfy the dynamics in \eqref{eq:subsystem1}. Further, to guarantee that those local solutions satisfying \eqref{eq:subsystemdyncons} stay on the same vector flow, i.e., belong to one unique solution, we employ an implicit stage-$3$ Runge-Kutta method for formulating this objective as an equality constraint. 
Concretely, we use 
Hermite interpolation to compute the interpolated value of $\xv_c^\id$ and its slope $\dot\xv_c^\id$ (see equation (30) of \cite{hereid_dynamic_TRO}) at the center of the subinterval $[t^\id,t^{\id+1}]$. 
Then the \defit{collocation constraints} are formed as: 
\par\vspace{-4mm}{\small 
\begin{align}
    d(\xv^\id, \xv^{\id+1}, u^\id_i) \defeq 
    \dot{\xv}^\id_c - F(\xv^\id_c, u^\id_i) = 0
    \tag{C.3}
\end{align}}%
\vspace{-6mm}\par\noindent

\newsec{Physical Constraints \& Periodic Constraints.} 
A set of inequality constraints (\textit{path constraints}) $ p(\vartheta^\id)\geq 0 $ are used to enforce conditions along the time horizon. For robotics, these are widely applied as obstacle avoidance condition, 
and some feasibility conditions for the dynamical system, representing real-world physics. 
In our application --- the walking dynamics of quadrupeds, the inequality constraints are used to define the friction cone condition and maximum ground clearance of the swing foot to be higher than 8 cm.

In addition, a set of equality constraints are imposed on the decision variables at $t=0,T$ to ``connect'' the initial and final condition: 
$b(\xv^0, \xv^\rmK) = 0,$
so that the optimal solution of the optimization is a periodic solution of the dynamical system. 
Particularly, the dynamics of quadrupedal locomotion include both continuous and discrete dynamics, forming a \textit{hybrid control system}. To find a periodic solution (ambling motion), we have the periodic constraint as: 
\par\vspace{-4mm}{\small 
\begin{align}
    b(q_i^0, \dot{q}_i^0, q_i^\rmK, \dot{q}_i^\rmK) = 
    \begin{bmatrix} \Delta(q^\rmK_i)\hspace{0.4mm} \dot{q}^\rmK_i -  \dot{q}_i^0 
    \vspace{0.5mm}\\ q_i^{\hspace{0.4mm}\rmK} - q_i^0  \end{bmatrix} =0
    \tag{C.6}
\end{align}
}%
\vspace{-4mm}\par\noindent 
where $\Delta(\cdot)$ represents the plastic impact dynamics that maps the pre-impact velocity $\dot q_i^\rmK$ to its post-impact term. 

\begin{figure*}[t!]
\vspace{1.3mm}
	\centering
	    \includegraphics[width=0.98\textwidth]{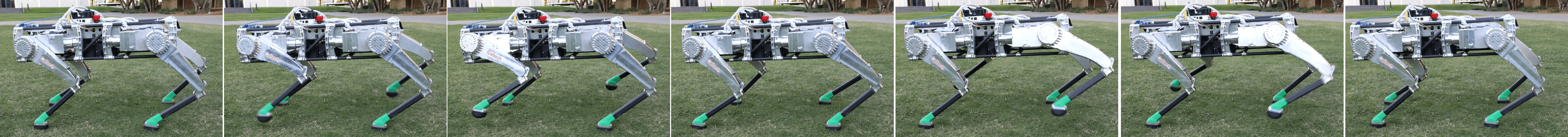} 
	    \vspace{1mm}\\
		\includegraphics[width=0.153\textwidth]{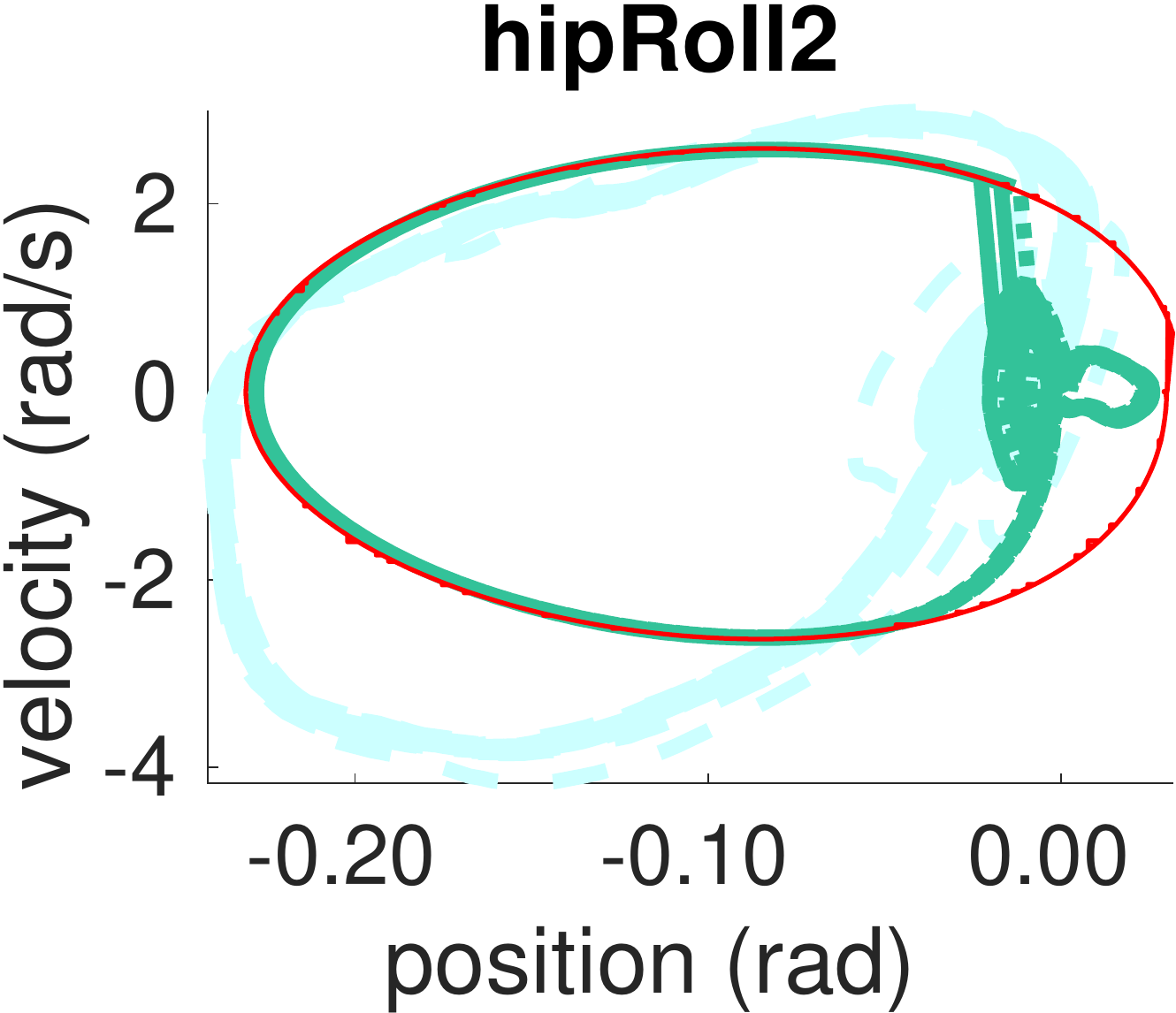}
		\includegraphics[width=0.153\textwidth]{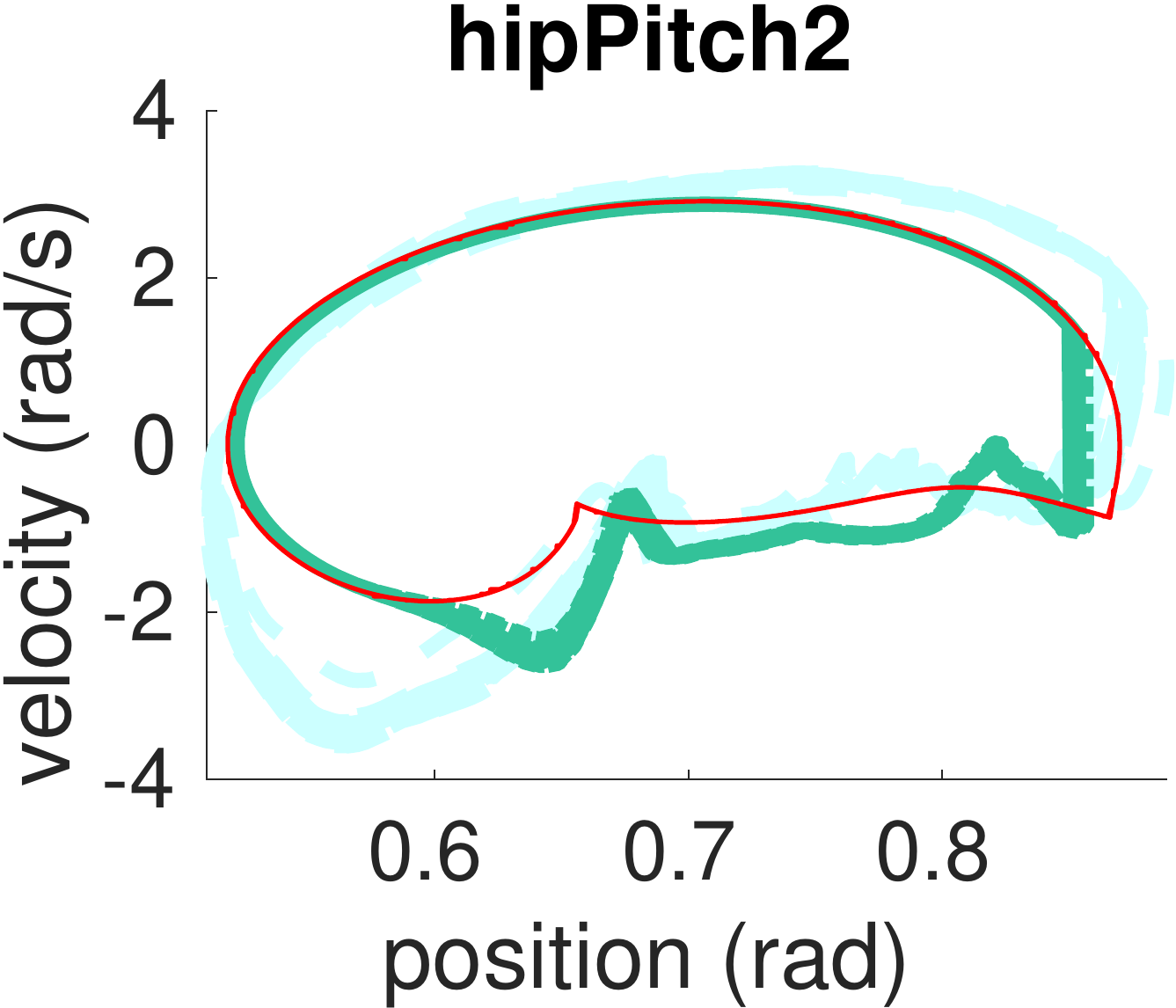}
		\includegraphics[width=0.153\textwidth]{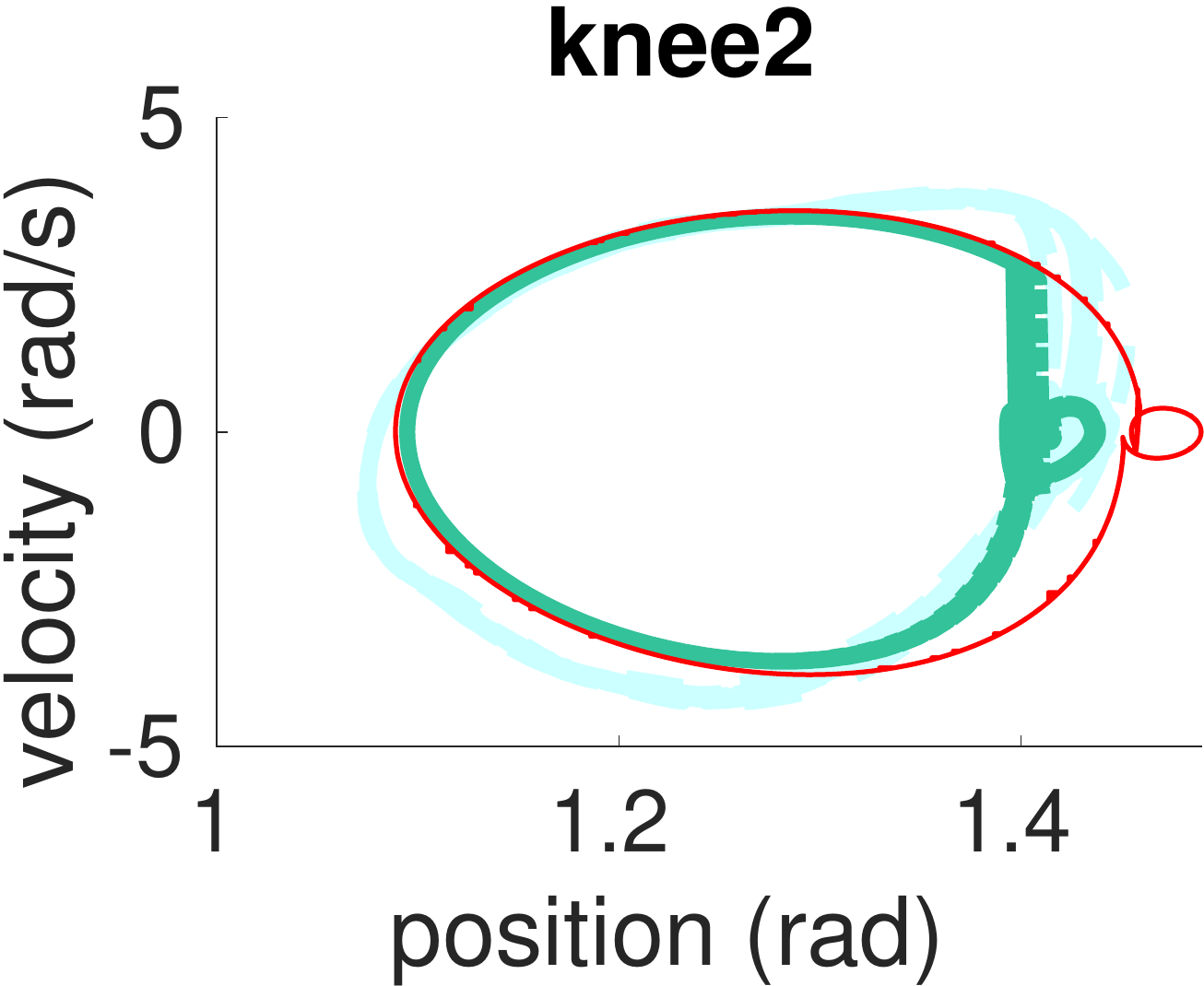}
		\quad
		\includegraphics[width=0.153\textwidth]{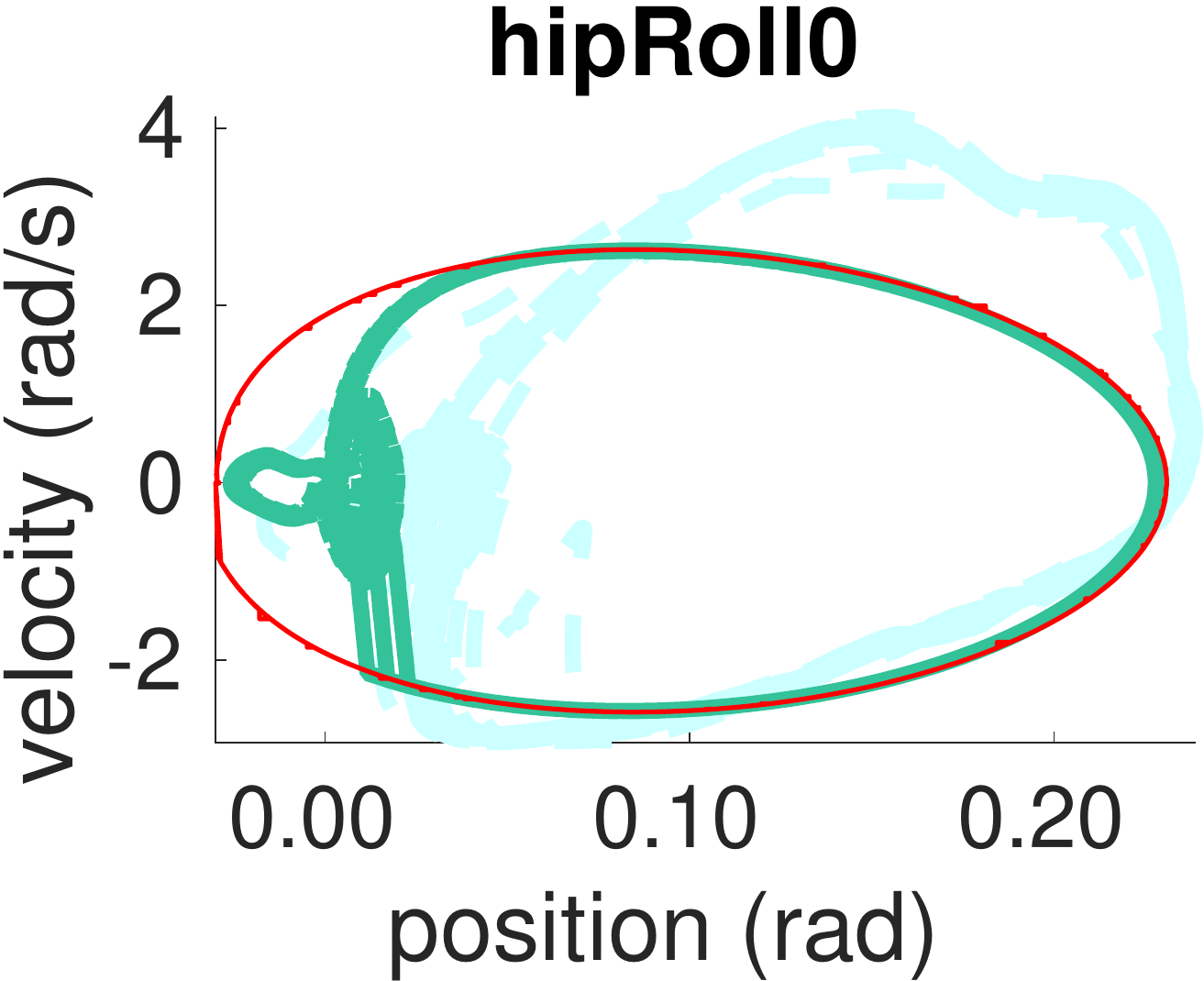}
		\includegraphics[width=0.153\textwidth]{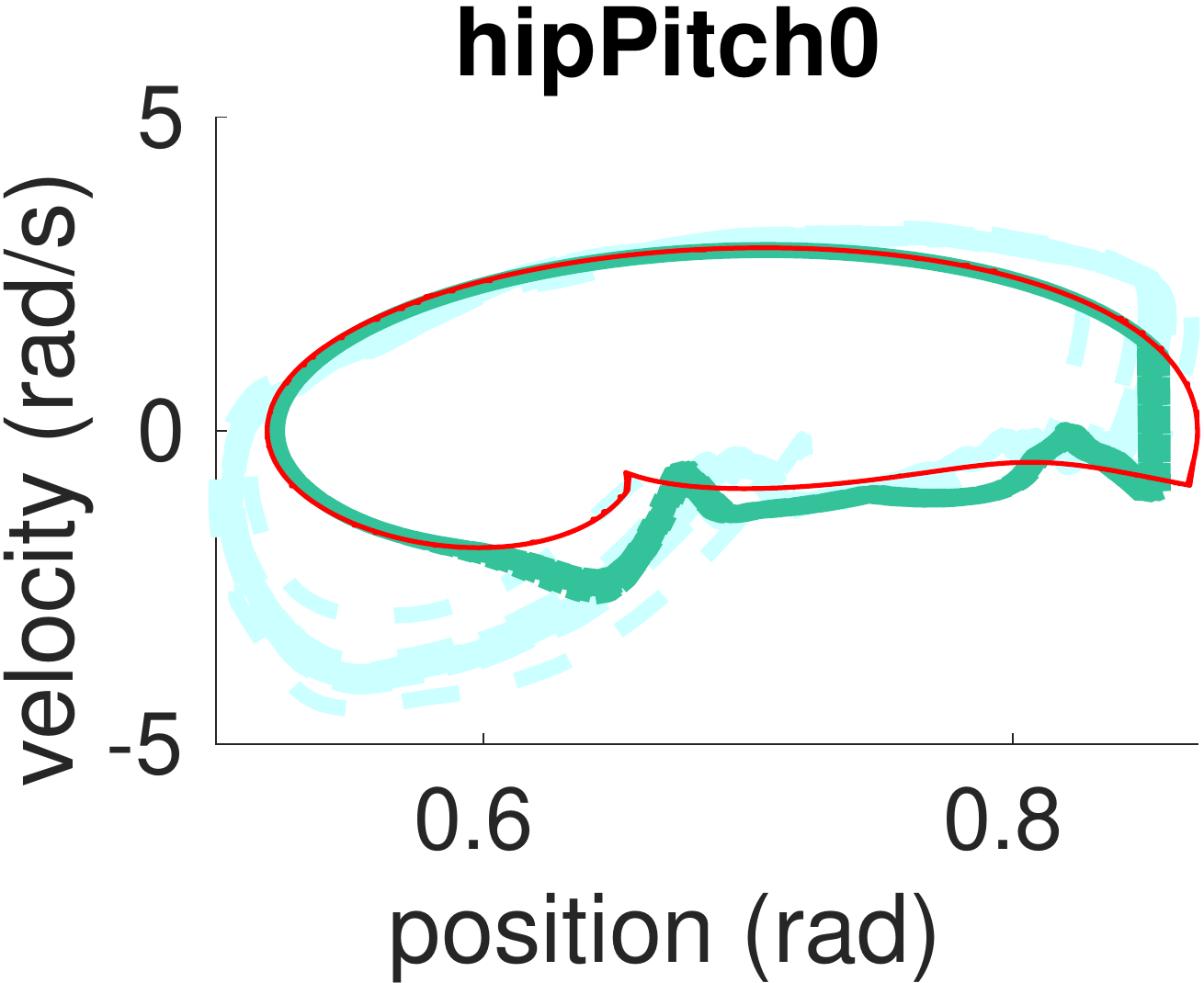}
		\includegraphics[width=0.153\textwidth]{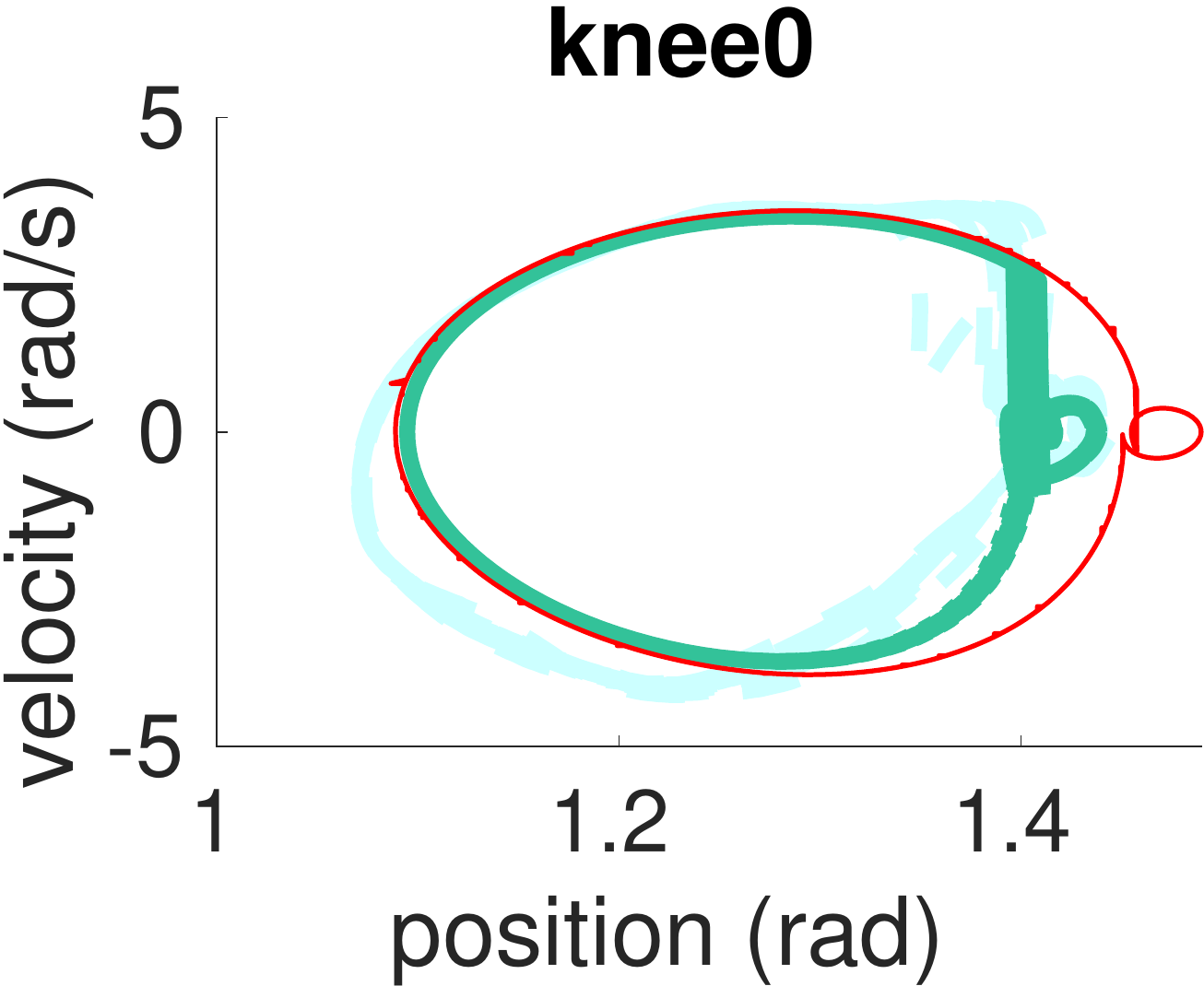}
	\vspace{-2.3mm}
    \caption{
    \textit{Top:} Snapshots showing a full step of the ambling gait in an outdoor lawn. 
    \textit{Bottom:} The periodic trajectory produced by optimization \eqref{eq:generalOPT} (in red) vs. the experimental tracking data (in cyan) vs. RaiSim simulation data (in green) in the form of phase portrait (limit cycle) using 18 seconds' data. 
    }
	\label{fig:walk2exp} 
	\vspace{-5mm}
\end{figure*}

\newsec{Optimization problem.} To find the periodic solution of dynamical system \eqref{eq:subsystem1}, we now parse this coupled controlling problem of the isolating \ith subsystem as:  
\par\vspace{-4.5mm}{\small
\begin{align} \label{eq:generalOPT}
    \underset{\mathbf{X}}{\operatorname{argmin}}\ \ &\Phi( \mathbf{X} ) & \tag{NLP}\\
    \text{s.t.\ \ }
    & F_{\text{zero}}(\vartheta^\id) = 0        &\id=0,1\ldots,\rmK   \tag{C.1}\\
    & F_{\text{dyn}}(\vi)= 0                    &\id=0,1\ldots,\rmK   \tag{C.2}\\
    & d(\xv^\id, \xv^{\id+1}, u^\id_i) = 0      &\id=0,1\ldots,\rmK-1 \tag{C.3}\\
    & \vartheta^\id \in \X\times\Z\times\Us     &\id=0,1\ldots,\rmK   \tag{C.4}\\
    & p(\vartheta^\id) \geq 0                   &\id=0,1\ldots,\rmK   \tag{C.5}\\
    & b(\mathbf{X}) = 0                         &                     \tag{C.6}
\end{align}}%
\vspace{-6mm}\par\noindent 
where $\Phi(\cdot)\in\R$ is the cost function. 
Here, we pick the cost function as the acceleration of the torso orientation to yield a less energetic motion for the ease of experiments. 
(C.4) defines the upper and lower bounds of the decision variables, i.e., that they live in the admissible space of values. In the application of walking, this was used to define the feasible configuration space and  the actuator torque less than 40N$\cdot$m. The other constraints are as stated as above.

\newsec{Solutions. } As a result, the optimization \eqref{eq:generalOPT} can simultaneously produce trajectories (solutions) of the states $\{ x_i(t), z(t)\}$, $u^{\Zbf}_j(t)$ that renders the zero dynamics manifold $\Zbf_j$ invariant and the open-loop controller $u_i(t)$,  $\forall t\in[0,T]$ for which these solutions are defined. 
Note that one can also enforce the dynamics $\dot x_i^\id + \varepsilon x_i^\id=0$ with $\varepsilon>0$ to guarantee the converging attribute of the \ith isolating subsystem, in which case the controller $u_i(x_i,z)$ is equivalently an input-output feedback linearization controller. 
Per Theorem \ref{thm:subsystem}, given $u^{\Zbf}_j$ that renders invariant $\Zbf_j$ and the feedback controller $u_i(x_i,z)$, we can compute $\lambdaZ(t)$ using \eqref{eq:lambdaeform}, hence $(\iota(x_i(t), z(t)),\lambdaZ(t))$ is a solution of the original CDS. 
Further, by imposing the periodic condition on the solution's boundary condition, the optimization produced a periodic solution of period $T$ to the CCS. Therefore, according to Corollary \ref{cor:periodic}, $(\iota(x_i(t),z(t),\lambdaZ(t))$ is a periodic solution of the CDS with period $T$.

\newsec{Application to quadrupeds.} 
When posing the control problem of quadrupeds, we leverage the subsystems representing the front and rear bipeds: $\bipedz{\rmf}$ and $\bipedz{\rmr}$, as given in \eqref{eq:bizieom}. Note that these subsystems are still coupled through $\lambda$---while this could be explicitly solved for via Lemma \ref{lemma:coupling},  
we keep it implicit due to the complexity of inverting the mass-inertia matrix for this particular robotic application. 
The \ith subsystem yield (C.1), (C.2) and (C.3) for \eqref{eq:generalOPT}. Specifically for all of the grid indices  $\id=0,1,...5$, we have the decision variables:  $\vi = \{q_\rmf^\id, \dot{q}_\rmf^\id, \xi_\rmr^\id, \dot\xi_\rmr^\id, u_\rmf^\id, F_\rmf^\id, F_\rmr^\id,  \alpha_\rmf,\\ \lambda_e^\id \}$. 
Finally, the optimization converged to a periodic solution of the isolated bipedal system, which can then be reconstructed to obtain the ambling motion of the quadrupedal robot (shown in \figref{fig:walk2exp}) according to Theorem 1. We report that the computation took 17.6s and 295 iterations of searching. Comparing to the traditional full-model based approaches \cite{ma2019First}, whose fastest record was 42s, the proposed method is $58\%$ faster.

To validate the proposed periodic orbit generation method using  \systems, we conducted experiments in indoor and outdoor environments, as well as in a physics engine ---RaiSim. 
The implemented controller is a time-based PD approximation of input-output linearizing controllers to track the the desired outputs (represented by $\alpha_\rmf, \alpha_\rmr = \mathcal{M}\alpha_\rmf$): 
\par\vspace{-5mm}{\small \begin{align*} 
    u_i(q_i, \dot{q}_i, t) =  
        -k_p\big(\dot y^a(q_i) - y^d_t(t, \alpha_i) \big) - k_d\big(y^a(q_i) - \dot y^d_t(t, \alpha_i) \big)  
\end{align*}}%
\vspace{-5mm}\par\noindent 
with $k_p, k_d$ the PD gains. 
The result is successful ambling in simulation, indoor (research lab) environment and outdoor rough terrains. See \cite{video} for the video and \figref{fig:walk2exp} for walking tiles and a comparison for the logged data with the optimized trajectory generated from \eqref{eq:generalOPT}. Importantly, we note that the averaged absolute torque inputs are 9.47, 6.45, 17.56 N$\cdot$m for the hip roll, hip pitch and knee motors, all of which are within the hardware limitations. Thus we are able to translate the theoretic results on gait generation for \systems\ to hardware in a physically realizable fashion.

\section{Conclusion} \label{sec:conclusion}

As inspired by robotic systems, this paper presented the new formulation of \systems: control systems that are connected via coupling relations and coupling inputs. We demonstrated how these systems can be reduced to a single subsystem that encodes the behavior of the full-order coupled system; this was achieved through leveraging zero dynamics and coupling relations. 
The main result of this paper was that solutions for these isolated subsystems are solutions for the full-order systems. 
Building on this, we constructed a nonlinear optimization problem on only a given subsystem that yields periodic orbits for the full-order dynamics.  
Finally, the application of these ideas were considered for \systems\ from which a specific example includes quadrupeds. 
%
%
This was demonstrated through experiments on hardware. An important future direction of this work is to expand the \system\ related concepts to system with more than two subsystems. 



\bibliographystyle{abbrv}
\bibliography{cite}

\begin{thebibliography}{10}

\bibitem{video}
Experimental video. \url{https://youtu.be/GlpgSXMinoU}.

\bibitem{ames14CLF}
A.~Ames, K.~Galloway, K.~Sreenath, and J.~Grizzle.
\newblock Rapidly exponentially stabilizing control lyapunov functions and
  hybrid zero dynamics.
\newblock {\em Automatic Control, IEEE Transactions on}, 59(4):876--891, 2014.

\bibitem{antonelli2013interconnected}
G.~Antonelli.
\newblock Interconnected dynamic systems: An overview on distributed control.
\newblock {\em IEEE Control Systems Magazine}, 33(1), 2013.

\bibitem{chung2009cooperative}
S.-J. Chung and J.-J.~E. Slotine.
\newblock Cooperative robot control and concurrent synchronization of
  {L}agrangian systems.
\newblock {\em IEEE transactions on Robotics}, 25(3):686--700, 2009.

\bibitem{featherstone2014rigid}
R.~Featherstone.
\newblock {\em Rigid body dynamics algorithms}.
\newblock Springer, 2014.

\bibitem{ganesh2007composition}
S.~Ganesh, A.~D. Ames, and R.~Bajcsy.
\newblock Composition of dynamical systems for estimation of human body
  dynamics.
\newblock In {\em International Workshop on Hybrid Systems: Computation and
  Control}, pages 702--705. Springer, 2007.

\bibitem{Grizzle2014Models}
J.~W. Grizzle, C.~Chevallereau, R.~W. Sinnet, and A.~D. Ames.
\newblock Models, feedback control, and open problems of {3D} bipedal robotic
  walking.
\newblock {\em Automatica}, 50(8):1955 -- 1988, 2014.

\bibitem{hamed2019hierarchical}
K.~A. Hamed, V.~R. Kamidi, W.-L. Ma, A.~Leonessa, and A.~D. Ames.
\newblock Hierarchical and safe motion control for cooperative locomotion of
  robotic guide dogs and humans: A hybrid systems approach.
\newblock {\em arXiv preprint arXiv:1904.03158}, 2019.

\bibitem{hereid_dynamic_TRO}
A.~Hereid, C.~M. Hubicki, E.~A. Cousineau, and A.~D. Ames.
\newblock Dynamic humanoid locomotion: A scalable formulation for {HZD} gait
  optimization.
\newblock {\em IEEE Transactions on Robotics}, pages 1--18, 2018.

\bibitem{ma2019First}
W.-L. Ma, K.~Akbari~Hamed, and A.~D. Ames.
\newblock First steps towards full model based motion planning and control of
  quadrupeds: A hybrid zero dynamics approach.
\newblock In {\em 2019 IEEE International Conference on Intelligent Robots and
  Systems (IROS)}, Macau, China, 2019.

\bibitem{ma2019bipedal}
W.-L. Ma and A.~D. Ames.
\newblock From bipedal walking to quadrupedal locomotion: Full-body dynamics
  decomposition for rapid gait generation.
\newblock {\em {arXiv preprint:1909.08560}}, 2019.

\bibitem{mesbahi2010graph}
M.~Mesbahi and M.~Egerstedt.
\newblock {\em Graph theoretic methods in multiagent networks}, volume~33.
\newblock Princeton University Press, 2010.

\bibitem{Murray1994mathematical}
R.~M. Murray, Z.~Li, S.~S. Sastry, and S.~S. Sastry.
\newblock {\em A mathematical introduction to robotic manipulation}.
\newblock CRC press, 1994.

\bibitem{reher2016realizing}
J.~Reher, E.~A. Cousineau, A.~Hereid, C.~M. Hubicki, and A.~D. Ames.
\newblock Realizing dynamic and efficient bipedal locomotion on the humanoid
  robot {DURUS}.
\newblock In {\em IEEE International Conference on Robotics and Automation
  (ICRA)}, 2016.

\bibitem{ren2005survey}
W.~Ren, R.~W. Beard, and E.~M. Atkins.
\newblock A survey of consensus problems in multi-agent coordination.
\newblock In {\em Proceedings of the 2005, American Control Conference, 2005.},
  pages 1859--1864. IEEE, 2005.

\bibitem{Sastry1999Nonlinear}
S.~Sastry.
\newblock {\em Nonlinear systems: analysis, stability, and control}, volume~10.
\newblock Springer New York, 1999.

\bibitem{Sreenath2011}
K.~Sreenath.
\newblock {A Compliant Hybrid Zero Dynamics Controller for Stable, Efficient
  and Fast Bipedal Walking on MABEL}.
\newblock {\em The International Journal of Robotics Research},
  30(9):1170--1193, Aug. 2011.

\bibitem{van2014port}
A.~van~der Schaft, D.~Jeltsema, et~al.
\newblock Port-hamiltonian systems theory: An introductory overview.
\newblock {\em Foundations and Trends{\textregistered} in Systems and Control},
  1(2-3):173--378, 2014.

\bibitem{Westervelt2007a}
E.~R. Westervelt, J.~W. Grizzle, C.~Chevallereau, J.~H. Choi, and B.~Morris.
\newblock {\em {Feedback Control of Dynamic Bipedal Robot Locomotion}}.
\newblock Control and Automation. CRC Press, Boca Raton, June 2007.

\end{thebibliography}


\end{document}